\documentclass{vldb}
\usepackage[ruled,vlined]{algorithm2e}
\usepackage{amsmath, MnSymbol, subfigure}

\makeatletter
\let\@copyrightspace\relax
\makeatother

\begin{document}
\newtheorem{theorem}{Theorem}
\newtheorem{corollary}{Corollary}

\title{Defining and Mining Functional Dependencies in Probabilistic Databases}

\numberofauthors{2} 
\author{
\alignauthor
Sushovan De\\
       \affaddr{Deptt. of Computer Science and Engineering}\\
       \affaddr{Arizona State University}\\
       \email{sushovan@asu.edu}
\alignauthor
Subbarao Kambhampati\\
       \affaddr{Deptt. of Computer Science and Engineering}\\
       \affaddr{Arizona State University}\\
       \email{rao@asu.edu}
}

\date{22 July 2010}

\maketitle

\begin{abstract}
Functional dependencies -- traditional, approximate and conditional are of critical importance in relational databases, as they inform us about the relationships between attributes. They are useful in schema normalization, data rectification and source selection. Most of these were however developed in the context of deterministic data. Although uncertain databases have started receiving attention, these dependencies have not been defined for them, nor are fast algorithms available to evaluate their confidences. This paper defines the logical extensions of various forms of functional dependencies for probabilistic databases and explores the connections between them. 
We propose a pruning-based exact algorithm to evaluate the confidence of functional dependencies, a Monte-Carlo based algorithm to evaluate the confidence of approximate functional dependencies and algorithms for their conditional counterparts in probabilistic databases.
Experiments are performed on both synthetic and real data evaluating the performance of these algorithms in assessing the confidence of dependencies and  mining them from data.
We believe that having these dependencies and algorithms available for probabilistic databases will drive adoption of probabilistic data storage in the industry.
\end{abstract}

%

\section{Introduction}
\label{sec-intro}

A lot of data generated today, especially that obtained from the web, is dirty, untrustworthy or uncertain. Yet we continue to store them in database engines that are ill-equipped to handle uncertainty. Handling uncertainty isn't a simple case of adding a `probability' attribute -- uncertain data has correlations, causations and query processing on such data is a probabilistic inference problem. It should be a top-priority for the database community to remove barriers that prevent data engineers from adopting probabilistic databases. Data obtained from the web becomes uncertain for a variety of reasons including the hardness of schema mapping, record linkage and presence of untrustworthy or dirty data. Traditionally, these problems have been tackled by taking the data one tuple at a time, and choosing the best possible alternative for it. Once the alternative is decided, the data is considered to be completely correct for the purpose for further analysis. However, probabilistic databases allow the data to fully reflect the different alternatives that were available. Further processing can then take into account all the alternatives, not just the most likely one. 

One of the barriers to using probabilistic databases is the lack of well-defined dependencies between attributes. In the case of traditional databases, such dependencies, in the form of exact and approximate functional dependencies, are used for both fast query processing and rectification of data. There are algorithms that will evaluate functional dependencies between attributes to aid in schema normalization~\cite{bernstein1976synth}, that will evaluate approximate functional dependencies (AFD), which helps filling out missing data in incomplete databases\cite{wolf2009query}. There are also conditional functional dependencies (CFDs), which help in cleaning and correcting data \cite{bohannon2007conditional}. However, such dependencies and algorithms are missing for probabilistic data. 

In this paper we extend these very useful dependencies so that they work with probabilistic databases in general. We generalize FDs to probabilistic functional dependencies (pFD), AFDs to probabilistic approximate functional dependencies (pAFD), and their conditional counterparts respectively to CpFD and CpAFD. We also investigate the relationship between these dependencies. In particular, we point out which of these dependencies are generalizations of, and hence subsume, others. We also provide fast algorithms for evaluating the confidence of these dependencies on probabilistic database, with a special focus on tuple-independent and tuple-disjoint independent databases \cite{Dalvi2007}. With the help of these algorithms, we describe how we can mine these dependencies from data, by using efficient methods to prune the search space of dependencies.

{\bf Motivating example:} 
Assume that two or more astronomers are observing and recording various objects in the sky, like in the Sloan Digital Sky Survey~\cite{sloanskysurvey}. They note various attributes of the objects, including the color, type, speed, frequency of oscillation, and position. Each observer may have his/her doubts about the data, so they may choose to enter alternatives as options in a probabilistic database. 
Such data is most naturally represented as a probabilistic database, where each tuple represents a different object in the sky and reflects the  curator's confidence in the observer as well as the observer's confidence in the data. Having represented this data, the curator can run a pAFD finding algorithm to discover dependencies that were as yet unknown, of the form $\left(color, speed \rightsquigarrow type\right)$. If the curator is aware of some dependencies that are expected to hold, he can verify their validity by running the appropriate dependency checking algorithm on this data.

\begin{figure}
\begin{center}
\includegraphics[keepaspectratio, width=250pt, clip, trim=0pt 0pt 0pt 0pt]{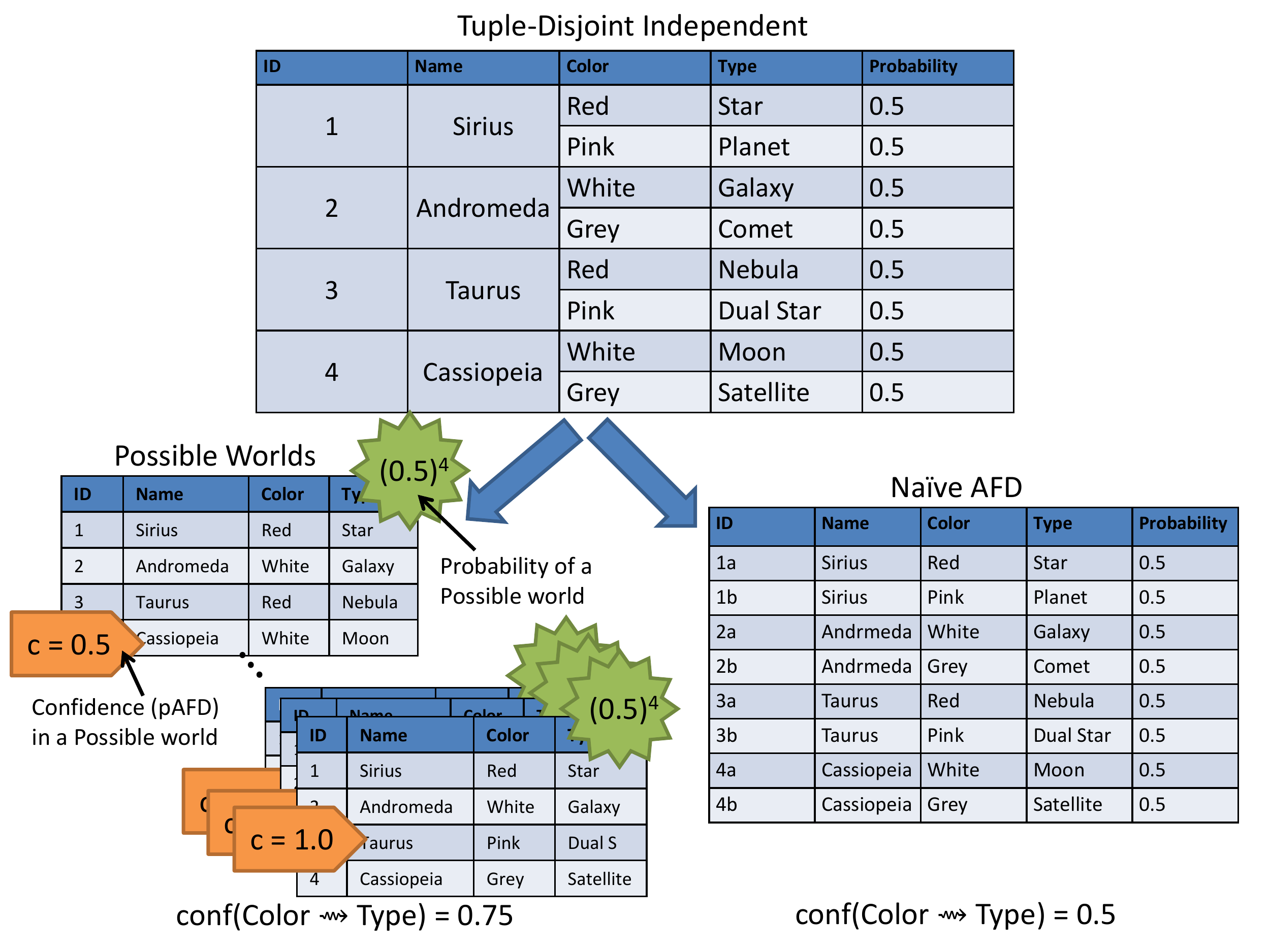}
\caption{Why pAFDs differ from a na\"ive interpretation of AFD. {\bf (top)} A tuple-disjoint independent database. {\bf (right)} An attempt to 
find the AFD na\"ively results in confidence of 0.5. {\bf (left)} The semantically correct value of 0.75.}
\vspace{-20pt}
\label{fig-pafd-importance}
\end{center}
\end{figure}

One pertinent question is whether these extensions are important or interesting enough to consider, or whether AFDs can be used directly. Pending empirical study, (see Section \ref{sec-real-data}), we demonstrate their importance with an example. Let us say that we have a probabilistic database as shown in Figure \ref{fig-pafd-importance}. We are interested in finding out whether or not the dependency Color $\rightsquigarrow$ Type has a high confidence. Strictly speaking, we cannot evaluate its AFD, because the concept of AFD does not apply to probabilistic data. However, we can apply what might be called an `intuitive' extension of an AFD to probabilistic data by considering each option as an independent tuple and weighing them according to their probabilities. If we do that, Figure~\ref{fig-pafd-importance} shows that this na\"ive method calculates a probability that is quite low, and different from the correct value dictated by probabilistic semantics.

The rest of the paper is organized as follows. We start by defining the various dependencies for probabilistic databases, and in the following section we explore the theoretical connections between them. In Section \ref{sec-algorithms} we propose some algorithms to evaluate the confidence of these quantities and show how to mine them, and in the section following that, show experiments that evaluate the effectiveness of these algorithms. We end with a discussion of related work and present our conclusions in Section \ref{sec-conclusions}.

\begin{figure}
\begin{center}
\includegraphics[keepaspectratio, width=200pt, clip, trim=0pt 80pt 0pt 0pt]{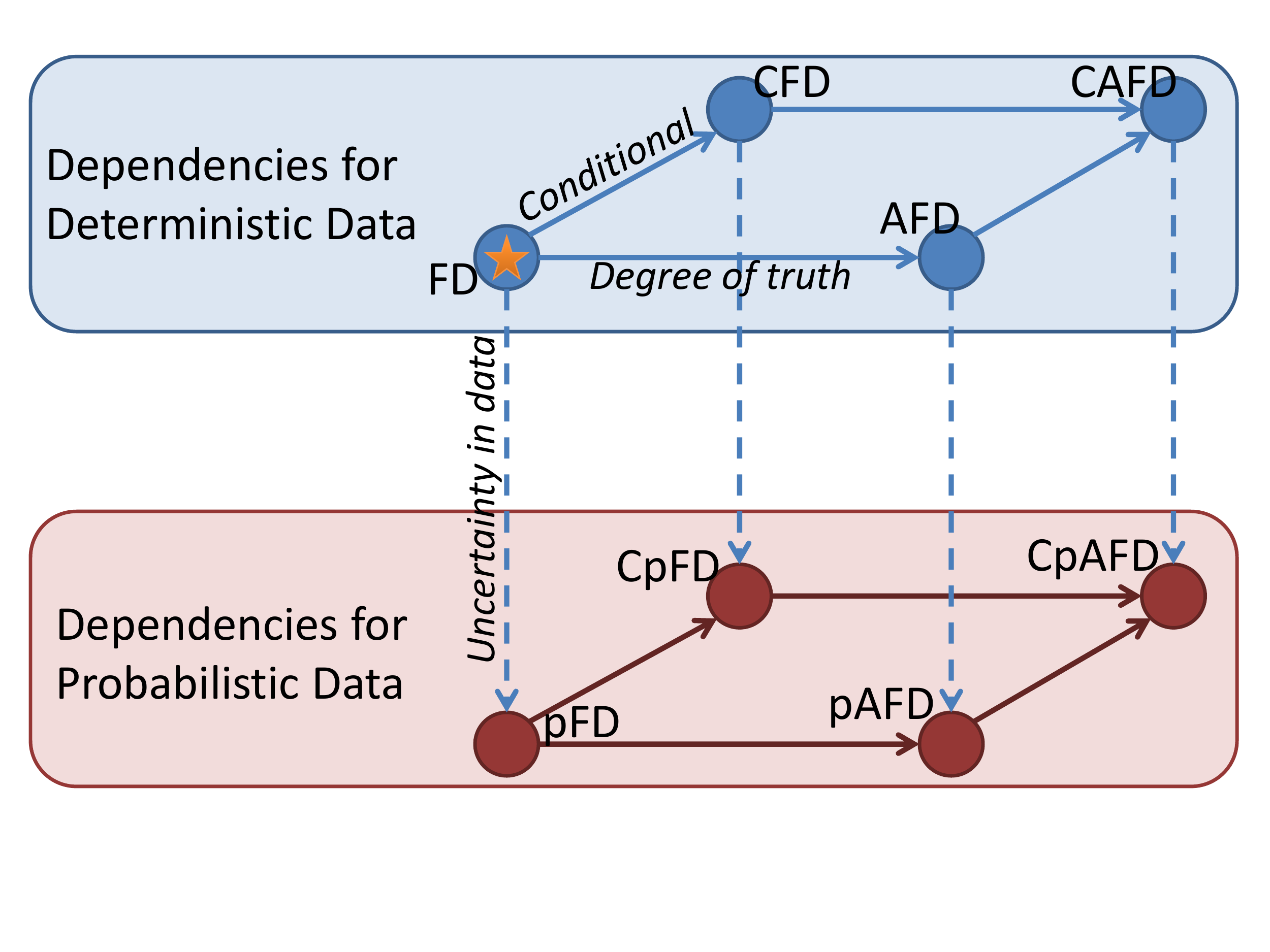}
\caption{The relationship between various dependencies. FDs are the least general, when extended by adding degree of truth, we get AFD; when extended to probabilistic databases, we get pFD, when extended to include conditional dependencies, we get CFDs. Combinations of these properties give rise to the other dependencies.}
\vspace{-20pt}
\label{fig-spectrum-dep}
\end{center}
\end{figure}
\section{Definitions}
\label{sec-definitions}

\subsection{Probabilistic Database}
\label{sec-pdb}
In this paper we follow the possible worlds model for a probabilistic database:  a probabilistic database is  a collection of possible worlds, with each possible world being a deterministic database and an associated probability. 
Figure \ref{fig-pafd-importance} shows a probabilistic database. The possible worlds representation is the set of relations on the bottom left. 


We denote a deterministic relation by symbol $R$, which has attributes $(A_1, A_2, ... , A_n)$. Sets of attributes are denoted by $X$ or $Y$. Uncertain relations are denoted by $D$, and each uncertain relation comprises possible worlds $(P_1, P_2, ... P_m)$, with  and attributes $(A_1, A_2, ... , A_n)$.

\subsection{Probabilistic Functional Dependencies (pFD)}
\label{sec-pfd}

Given a deterministic relation $R$ a functional dependency (FD) is defined as $(X \rightarrow Y)$, where $X$ and $Y$ are sets of attributes. The FD is said to hold if whenever two tuples share the same values of $X$, they have the same values of $Y$.

We can generalize this idea to probabilistic databases, as shown in Figure \ref{fig-spectrum-dep} by the `uncertainty in data' axis. Given an uncertain relation $D$, a probabilistic functional dependency, pFD, is defined as $(X \rightarrow Y)$. The pFD is associated with a quantity called its confidence which is the fraction of possible worlds in which the corresponding FD holds. 

Consider the possible world representation in  Figure \ref{fig-pafd-importance}. In the first possible world, the data in tuple 1, {\it (Red, Star)} conflicts with tuple 3, {\it (Red, Nebula)}. As a result the contribution of that possible world towards the confidence of the pFD $Color \rightarrow Type$ is zero. The last possible world in the figure shows a non-zero contribution. The FD holds within the world, hence the entire probability of the world is added to the pFD confidence score.

It should be noted that pFDs suffer from the same kind of flaws that traditional FDs do. If the data is dirty, just a few tuples that do not conform to the pFD might cause an entire possible world to be not counted. We can address these concerns with pAFDs.

\subsection{Probabilistic Approximate Functional \\Dependencies (pAFD)}
\label{sec-pafd}
AFDs generalize FDs by adding the concept of a `degree of truth' to an FD, which is also illustrated in Figure \ref{fig-spectrum-dep}. Given a deterministic relation $R$, an approximate functional dependency AFD is defined as $(X \rightsquigarrow Y)$, where $X$ and $Y$ are sets of attributes, with $X$ approximately determining $Y$. The confidence of an AFD may be defined in various ways. Following  \cite{huhtala1999tane},  we define the AFD confidence as as one minus the minimum fraction of tuples that need to be removed from the relation for the FD to hold.

In an uncertain relation $D$, a probabilistic approximate functional dependency, pAFD, is defined as $(X \rightsquigarrow Y)$. The confidence of the pAFD is the expected confidence of the AFD over the possible worlds, (i.e. average of the confidence of the corresponding AFD in each possible world weighted by the probability of that world). 

For example, in Figure \ref{fig-pafd-importance}, in the first possible world, the data in tuple 1, {\it (Red, Star)} conflicts with tuple 3, {\it (Red, Nebula)}. As a result only one of them can be considered as contributing towards the AFD within that possible world. A similar argument holds for tuples 2 and 4. We can therefore say that among the four tuples in the possible world, two support the AFD, and two don't; hence the confidence of the AFD is 0.5. 
In situations where data is both uncertain and noisy, pAFDs are useful 
for judging  the relationship between attributes.

\subsection{Conditional Probabilistic Functional\\ Dependencies (CpFD)}
\label{sec-cpfd}

We can extend functional dependencies in yet another way -- by making them conditional on specific values of the data. Given a deterministic relation $R$, a CFD is a pair $(X \rightarrow Y, T_p)$. $X \rightarrow Y$ is a standard FD. $T_p$ is the pattern tableau, which is a relation with the attributes $(X \cup Y)$. The semantics of a CFD \cite{bohannon2007conditional} are as follows: A CFD holds on R if the corresponding FD holds on the subset of tuples that {\it matches} the pattern tableau. Every tuple in $T_p$ is either a constant or the wildcard character (\_). A constant $a$ in $T_p$ matches only the constant $a$ in $R$. The wildcard `\_' in $T_p$ matches any value in the real tuple. For a pair of tuples to violate the CFD, they must agree on every attribute in $X$ but different values for some attributes in $Y$, and the set of attributes $X \cup Y$ must match the pattern tableau.



We can extend this concept to probabilistic databases. Given an uncertain relation $D$, a conditional probabilistic functional dependency, CpFD, is the pair $(X \rightarrow Y, T_p)$, where $X$ and $Y$ are sets of attributes and $T_p$ is the pattern tableau. The confidence of a CpFD is the fraction of possible worlds where the corresponding CFD holds.


A CpFD is an extension of the concept of CFD to uncertain databases, but it does not tolerate dirty data. If even one tuple in a possible world violates the CpFD, the entire possible world probability is not counted.

\subsection{Conditional Probabilistic Approximate \\Functional Dependencies (CpAFD)}
\label{sec-cpafd}

Given a deterministic relation $R$, a conditional approximate functional dependency, CAFD, is defined as the pair $(X \rightsquigarrow Y, T_p)$. The confidence of the CAFD is one minus the fraction of tuples that need to be removed from the subset of tuples that match the pattern tableau $T_p$ such that the CFD $(X \rightarrow Y)$ holds.

CAFDs do support a fractional confidence value, thus they can be used where the data is expected to be noisy and when the dependency holds only conditionally. Therefore it is useful to extend for probabilistic data. Given an uncertain relation $D$, a conditional probabilistic approximate functional dependency (CpAFD) is the pair $(X \rightsquigarrow Y, T_p)$. Its confidence is the weighted average of the confidence of the corresponding CAFD in each possible world, weighted by the probability of that possible world. 

A CpAFD extends the notion of a functional dependency in the most general way among all of the dependencies discussed in this paper. It supports fractional confidence values, possible world semantics, as well as operating on a select part of the database.


\section{Relationships among dependencies}
\label{sec-relationships}

The dependencies defined in Section \ref{sec-definitions} are all generalizations of FDs, as shown in Figure \ref{fig-spectrum-dep}. It is natural therefore, that we can express the less general dependencies in terms of the more general ones, and that we can induce relationships among them. That is what we will attempt to do in this section.

An FD is an AFD of confidence 1. AFDs allow dependencies that hold approximately, thus they extend FDs along the `degree of truth' axis as shown in Figure \ref{fig-spectrum-dep}. It should be noted that the confidence of an AFD can never be zero. This is because, in a relation $R$ with at least one tuple, even if all different values of $Y$ occur with the same values of $X$, we can remove all tuples except for one for each different set of values of $X$ and have a non-zero set of tuples left in our database. Thus, the confidence of the AFD, which is is one minus the fraction of tuples removed, is non-zero.  We can make a similar comparison between pFDs and pAFDs. 

\begin{theorem}
\label{th-pafd-greater-pfd}
The confidence of a pAFD is always larger than the confidence of the corresponding pFD.
\end{theorem}
\begin{proof}
In every possible world that the pFD holds, the confidence of the pAFD is 1. In every possible world that the pFD does not hold, the confidence of the pAFD is non-zero. Thus the confidence of the pAFD, which is the weighted sum of the confidences in each possible worlds, would be greater than that of the pFD.
\end{proof}

However, the reverse is not true. The information that a pAFD holds with a very high confidence does not imply that the pFD will have a high confidence, in fact, the confidence of the pFD may even be zero, as there might be a few tuples in every possible world that do not conform to the pFD.



Conditional dependencies extend each of the previous dependencies so that they can be specified over only a part of the data. This is illustrated in Figure \ref{fig-spectrum-dep} by the `conditional' axis. However, the generalization to conditional dependencies can introduce inconsistencies. For example, if we know that an FD holds on a relation, the corresponding CFD is not guaranteed to hold, since the tableau could introduce impossible cases \cite{bohannon2007conditional}. For example, in a CFD $(A \rightarrow B)$, if the pattern tableau has two tuples, one requiring the value of $B$ to be $b$, the other requiring the value to be $c$, the CFD is clearly inconsistent, and no relation can satisfy it. However, in normal use cases, where the tableau of the CFD has been induced from the data, or else calculated in a non-malicious manner, intuitively we can state that if the FD holds, the CFD will also hold. The converse, however, is not true. If a CFD with a non-trivial tableau holds over a relation $R$, then we cannot guarantee that the corresponding FD holds.

A pattern tableau may select a non-zero number of tuples from  some possible worlds but no tuples from others. In the special case where the tableau selects some tuples from each possible world, we can state that the confidence of the CpFD and the confidence of a CFD are equal. More generally, we can state the following theorem:

\begin{theorem}
\label{th-cpfd-generalizes-pfd}
If a pFD $(X \rightarrow Y)$ holds over an probabilistic relation $D$ with confidence $p$, then the CpFD $(X \rightarrow Y, T_p)$ also holds over $R$ with a confidence not less than $p$, provided the CpFD is not inconsistent. 
\end{theorem}
\begin{proof}
Suppose the pFD holds on a certain possible world of $D$. The pattern tableau would then cause certain tuples to be eliminated from consideration. Even after elimination of a few tuples, the pFD will continue to hold, by definition. Thus that possible world will contribute towards the confidence of the CpFD. On the other hand, if there is a possible world in which the pFD does not hold, it is possible that those tuples that cause the pFD not to hold will be eliminated by the pattern tableau, causing the possible world to contribute to the confidence of the CpFD. So the fraction of possible worlds contributing to the CpFD is not smaller than the fraction contributing to confidence of the pFD.
\end{proof}

The same argument does not hold for pAFD and CpAFDs. In the case of CpAFDs, elimination of certain tuples may reduce its confidence. For example, consider the  CpAFD $(A \rightsquigarrow  \{B, C\}, T_1)$ where $T_1$ consists of the tuple $(\_, b_2, \_)$. If one of the possible worlds, $P$, of relation $D$ contains a 50 tuples of $(a_1, b_1, c_1)$ and 50 tuples of the form $(a_1, b_2, c_2),$ $(a_1, b_2, c_3)$ \ldots $(a_1, b_2, c_{51})$  the pAFD would have confidence 0.50, but the CpAFD would have confidence 0.02 (after eliminating the 50 tuples not matching $T_1$).

\section{Assessing and mining probabilistic dependencies}
\label{sec-algorithms}
We consider two problems in the presented framework:

{\bf Evaluating confidence:} Given a relation $R$, and a specified dependency find the confidence of the dependency. 

{\bf Mining dependencies:} Given a relation $R$, find a minimal set of dependencies 
that is equivalent to or more general\footnotemark than any set of dependencies that holds over $R$ with a confidence higher than a given threshold. 

\footnotetext{A set of dependencies $\Sigma_1$ is said to be more general than another set $\Sigma_2$ if every dependency in $\Sigma_2$ can be inferred from $\Sigma_1$ using Armstrong's axioms and the inference rules for conditional dependencies in \cite{bohannon2007conditional}.}

In the following sections we focus on evaluating the confidence of the dependencies, since it is the first step towards mining them. Once we have fast algorithms for evaluating the confidence, we then use methods in \cite{wolf2009query} to prune the space of dependencies we have to search through in order to mine them.

While we're interested in computing the confidence for any probabilistic database, we shall see that in the very general case, evaluation is exponential. So we also consider special cases, mainly focusing on tuple-disjoint independent (TDI) databases \cite{Dalvi2007}, which are a popular special case of a probabilistic database, in which every tuple with distinct keys is independent. We can think of this as a set of uncertain relations, where each tuple has a set of ``options" with each option having a probability. The decision about which option to pick for each tuple is taken independently. This significantly reduces the types of uncertainty and correlations that can occur among the tuples, but also makes many operations on the database tractable. These algorithms also work for Tuple-independent databases (TI), where each tuple has an existential probability, but does not have any options. The straightforward adaptation of these algorithms to TI databases is explained in the Appendix \ref{apx-algo-for-ti}.

\subsection{Assessing the confidence of a pFD}
\label{sec-algo-pfd}


\begin{figure}
\begin{center}
\includegraphics[keepaspectratio, width=250pt, clip, trim=0pt 140pt 0pt 0pt]{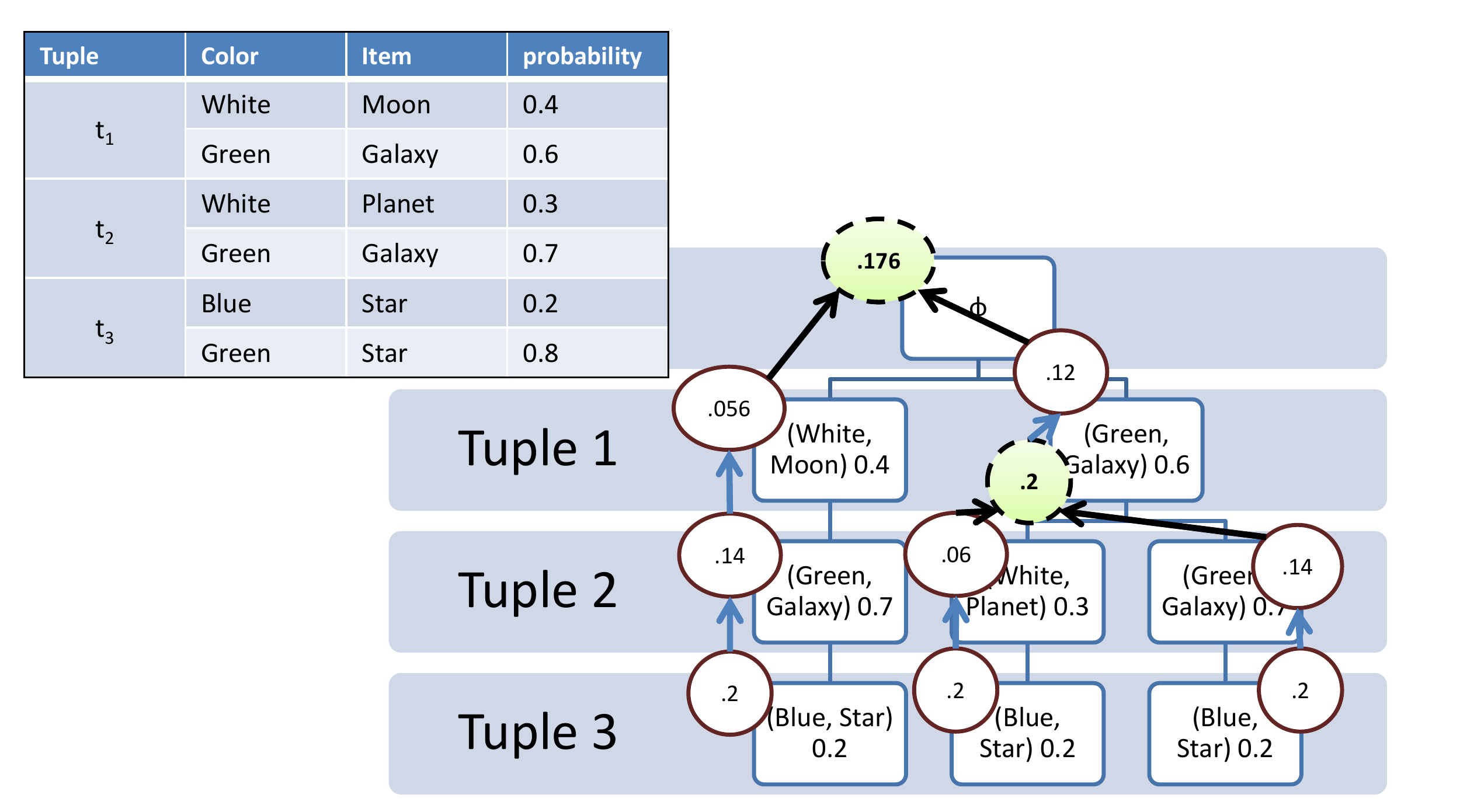}
\caption{The algorithm for computing the confidence of pFD. The dotted circles represent adding the probabilities from the child branches, the solid circles represent multiplication of probabilities of the child and the parent.}
\vspace{-20pt}
\label{fig-tree-method}
\end{center}
\end{figure}

In a generalized probabilistic database, represented by its possible worlds, finding the pFD would be polynomial in the combined size of the possible worlds, which is typically exponential in the number of entities it represents. For a more compact representation of a probabilistic database like TDI or TI, a na\"ive evaluation of the confidence of a pFD in a probabilistic database would likely take an exponential time in the number of tuples, since we would have to effectively generate the possible worlds and add up the probability of those in which the corresponding FD holds. Ordinarily, we would use Monte Carlo to sample the possible worlds (such as we will employ later to find the confidence of pAFDs), however, that approach does not work very well for pFD, since a single option that violates the dependency can bring the contribution of the entire sample to zero.

We now present an efficient algorithm that finds the confidence of a pFD in a tuple-disjoint independent database. This algorithm is exact and has the  property of being  exponential only in the {\em cardinality of the domain} of the attributes, rather than the number of tuples. In practice our algorithm finds useful probabilistic functional dependencies very efficiently, as most desirable pFDs have low specificity.

The complexity of the algorithm can be analyzed in terms of specificity is defined as the support of the association rules that make up the dependency \cite{aravind2008thesis}. Specificity is high when the association rules have a very low support, and the rule becomes less valuable. For example, a rule with high specificity such as {\it (Social Security Number $\rightarrow$ Color of hair)} will definitely hold, since SSN is a key, but is not very useful. On the other hand, an FD that states {\it (Zip Code $\rightarrow$ Street Name)} for addresses in England is a useful one. Each zip code appears multiple times in the database and the dependency gives us useful semantic information even though zip code is not a key.
For a low specificity pFD, the number of values a particular attribute can take is much less than the number of tuples, which makes our algorithm run more efficiently. A formal definition of specificity and its adaptation to TDI databases is presented in the Appendix \ref{apx-specificity}.



The algorithm exploits the fact that we are using a tuple-disjoint independent database. It keeps track of a set of association rules that comprise the pFD at any point in the algorithm. We pick the tuples one by one, and treat them independently. We next optimize the calculation of the pFD using two pruning criteria. First, if a particular option does not comply with the current set of rules, then the entire set of possible worlds that include that option will contribute zero confidence for the pFD. So we can terminate that branch right away. Second, all the options in the tuple comply with the ruleset, then the confidence of the pFD does not change whether or not we pick that tuple (it's contribution is 1). We can therefore ignore the tuple.

We can express the evaluation of this algorithm as a tree, see Figure \ref{fig-tree-method}. The tree branches whenever more than one option in a tuple is consistent with the current set of rules. Using the two criteria in the previous paragraph, we can choose an optimal order in which to pick the tuples so that the expression tree has the minimum width. We can then prove that the algorithm is exponential only in the cardinality of the domain of the attributes.


Once the execution reaches the leaf node, we track back. At each stage, we sum up the probability of all the branches that originated at that point. Then we multiply the result with the probability of the parent to compute the contribution of this branch to the confidence of the pFD.

\begin{algorithm}[t]
\DontPrintSemicolon

\SetKwData{Left}{left}\SetKwData{This}{this}\SetKwData{Up}{up}
\SetKwFunction{Union}{Union}\SetKwFunction{FindCompress}{FindCompress}
\SetKwInOut{Input}{input}\SetKwInOut{Output}{output}

\Input{A TDI database $D$, a pFD $P$ and the current set of rules $R$}
\Output{The confidence of $P$ in $D$}

\Begin{
$P \longleftarrow 0$\;
$T = ChooseBestRemainingTuple(D, R)$ \;
\If{$EntireTupleisCompatible(R)$}{
	\Return{ $ FindPfdRecursive( D \setminus T, R )$}\;
}
\For{Options $O \in T$} {
	\If{IsCompatible(O, R)}{
		$AddRule(O, R)$\;
		$P \longleftarrow P + Prob.(O) \times FindPfdRecursive( D \setminus T, R )$\;
		$RemoveFromRule(O, R)$\;
	}
}
\Return{P}\;
}
\caption{FindPfdRecursive\label{eq-tree-pfd}}
\end{algorithm}

The algorithm is formally presented as Algorithm \ref{eq-tree-pfd}. The running time of the algorithms is improved by the introduction of the function $ChooseBestRemainingTuple$, which picks from the remaining tuples those that do not cause branching in the evaluation tree. There are three kinds of such tuples -- those that completely comply with the current set of rules (their contribution is 1), or those that have no options that comply with the current set of rules (the branch is immediately terminated), or those that have only one option that complies with the current set of rules (there is no branching).

\subsection{Assessing the confidence of a pAFD}
\label{sec-algo-pafd}
It was possible for us to considerably speed up the calculation of the confidence of a pFD because as soon as an execution branch violated the association rules for that branch, we were able to terminate it. However, this technique does not work for calculating the confidence of a pAFD. During the execution of a similar algorithm for pAFD, if a tuple violates the majority association rule, it does not terminate the branch -- it merely reduces the confidence within that possible world. In addition to this, the association rules themselves might change once enough counterexamples are observed. As a result the algorithm becomes exponential time for a pAFD. We can, however, attempt to calculate the confidence of a pAFD approximately. 

{\bf Deterministic approximation:} One of the ways in which a pAFD may be approximated is to first create a deterministic database by completely ignoring the probabilities and treating every uncertain option as a tuple of a deterministic database. Obviously, this will violate the semantics of disjoint possible worlds. We then find the AFD over this deterministic database, and report the confidence of that AFD as the confidence of the pAFD.
\smallskip

{\bf Unioned approximation:} A better approximation would be to create a tuple independent (TI) database by taking every uncertain option in the database along with its probability and making it a tuple of the TI database. This causes us to lose all the correlation between the options of the same tuple in a TDI database, and we treat them as independent. Since this is effectively creating a union of all the options, we call this approach the unioned approximation to finding the confidence of a pAFD. We then find the pAFD over this tuple-independent database. Having the tuples completely independent of each other makes finding the confidence much easier. To find the confidence of the dependency $(X \rightsquigarrow Y)$, we can find all the association rules $(x, y)$ in the database. For each distinct value of $X$, we find that value $y_{max}$ of $Y$ which has the maximum support in the database. The pAFD value is given by $\sum{support(y_{max})}/\sum{support(x)}$.
\smallskip

{\bf Monte Carlo:} In order to maintain intra-tuple correlations, we can sample a subset of the possible worlds, and compute the confidence of the pAFD over that subset. We can then scale it up appropriately to find the confidence of the pAFD over the entire relation. It is clear that the more representative a subset we sample, the more accurate our pAFD computation will be. To choose a well-represented subset of possible worlds, we use a Monte Carlo simulation.

It is worth noting that the Monte Carlo technique does not make any assumptions about the Probabilistic Database under consideration, specifically, {\em it is not restricted to TDI databases}. For the case of the TDI database, we take one tuple at a time. For every tuple, we generate a random number to choose which option is to be picked. This is done proportional to the probability of the option. Since by  definition, different tuples are independent of each other, this results in generating a Monte Carlo sample of the TDI database. We find the AFD confidence of the sampled possible world, and weigh it with the probability of the possible world. We repeat this process till the the weighted average of the AFD values observed converges. That is reported as the confidence of the pAFD.
\smallskip

\begin{figure*}
\begin{minipage}[b]{0.45\linewidth}
\centering
\includegraphics[keepaspectratio, width=240pt, clip, trim=0pt 140pt 120pt 0pt] {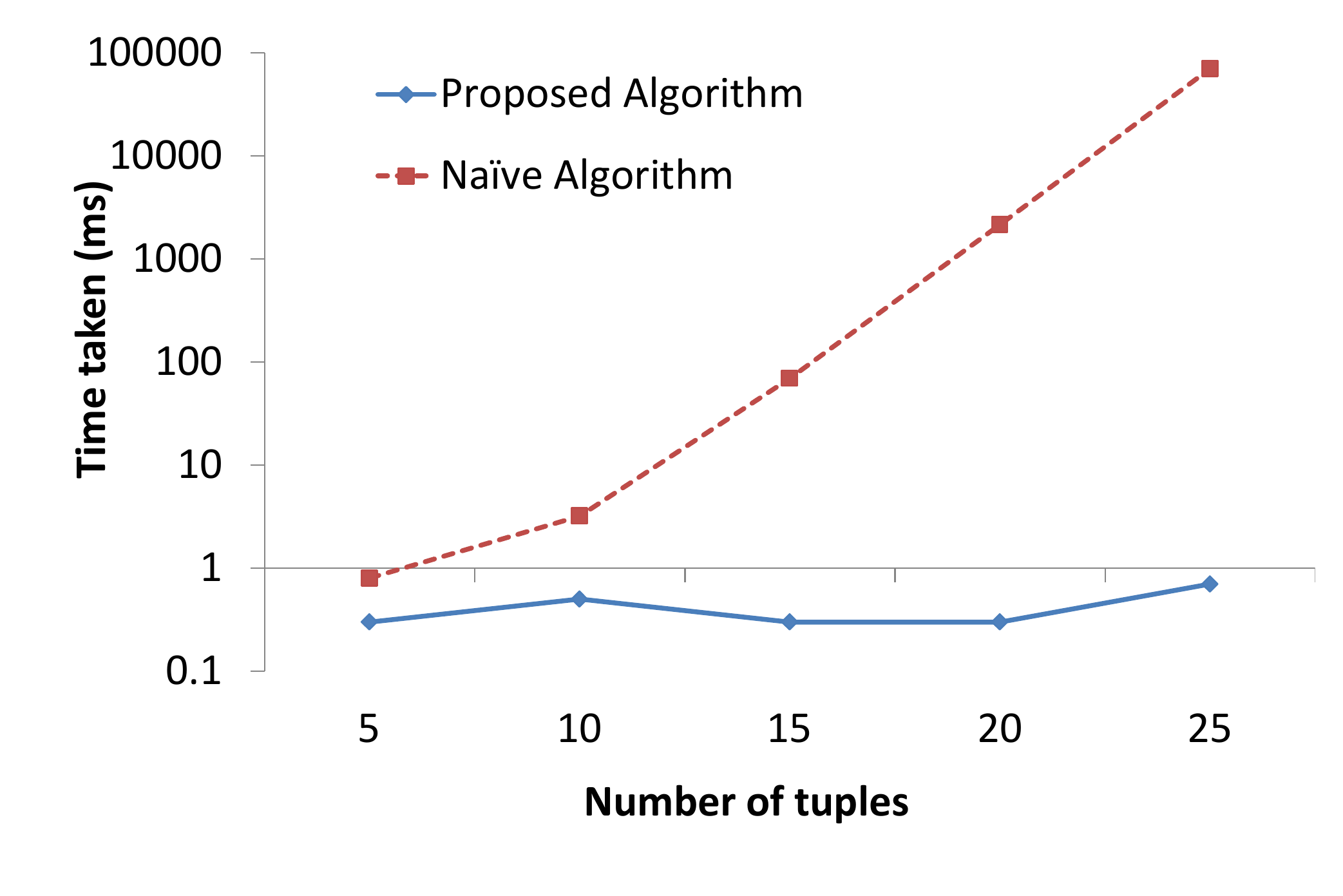}
\vspace{-0.4in}
\caption{Comparison between the time taken by the na\"ive algorithm and the proposed algorithm for the confidence of a pFD on a log-scale.}
\label{fig-pfd-naive-vs-prune}
\end{minipage}
\hspace{20pt}
\begin{minipage}[b]{0.45\linewidth}
\centering
\includegraphics[keepaspectratio, width=240pt, clip, trim=0pt 140pt 120pt 0pt] {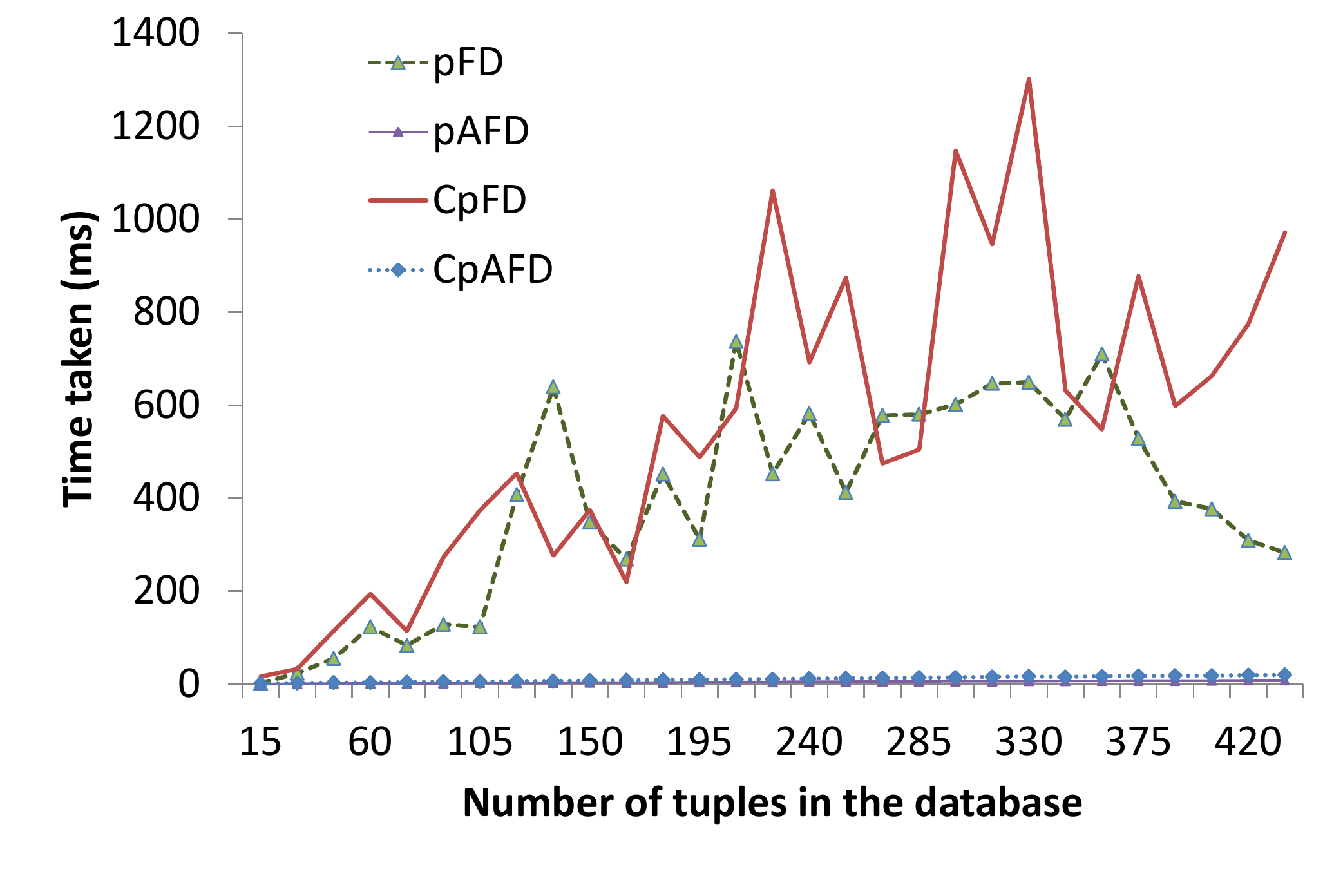}
\vspace{-0.4in}
\caption{ Comparison of time taken for the algorithms for various dependencies to run vs the number of tuples.}
\label{fig-synth-time-vs-tuples}
\end{minipage}
\end{figure*}

\subsection{Assessing the confidence of a CpFD}
\label{sec-algo-cpfd}

Our strategy to find the confidence of the conditional dependencies -- the CpFD and the CpAFD -- is a simple one. We first select the tuples from the database that match the pattern tableau, and then we run our corresponding pFD or pAFD on the resulting relation.

As Dalvi and Suciu show in \cite{dalvi2007efficient}, query processing on a probabilistic database can be a \#P-hard problem. However, there are certain queries that are guaranteed to have {\it safe-plans} which can be evaluated in polynomial time. Fortunately, selecting the tuples matching a pattern tableau is a safe query, and can be efficiently evaluated using the algorithms in \cite{dalvi2007efficient}.

We follow Bohannon {\it et. al.} \cite{bohannon2007conditional} for the query to find tuples that do not match the pattern tableau, appropriately modified for a probabilistic relation. Given a probabilistic relation $R$ with attributes $A_1, ... A_n$, and a CpFD $(X \rightarrow Y, T_p)$, we can use the following query to find the probabilistic options of the tuples that do not match the pattern tableau $T_p$ and hence can be removed from consideration:
\begin{align*}
&\,&\mathbf{select } &\; t \; \mathbf{from } \;R\, t, T_p\, t_p \\
&\,&\mathbf{where } &\; \textsc{not}( t[X_1] \asymp t_p[X_1] \,\textsc{and} ... \textsc{and} \,t[X_n] \asymp t_p[X_n])
\end{align*}
Here, $t[X] \asymp t_p[X]$ represents the condition that either $t[X] = t_p[X]$ or $t_p[X] = \_$. We replace all these tuples with the special symbol $\circledcirc$. Then we run the following query to find all tuples that violate the pattern tableau:
\begin{align*}
&\,&\mathbf{select } &\; t \; \mathbf{from } \;R\, t, T_p\, t_p \\
&\,&\mathbf{where } &\;  t[X_1] \asymp t_p[X_1] \,\textsc{and} ... \textsc{and} \,t[X_n] \asymp t_p[X_n] \, \textsc{and}\\
&\,&\, &\; (t[Y_1] \nasymp t_p[Y_1] \, \textsc{or} \, ... \, \textsc{or} t[Y_n] \nasymp t_p[Y_n]) 
\end{align*}
Here $t[Y] \nasymp t_p[Y]$ represents the condition that both $t[Y] \neq t_p[Y]$ and $t_p[Y] \neq \_$. We replace these options with the special symbol $\phi$ to denote that it violates the tableau. The main difference from \cite{bohannon2007conditional} in finding mismatches is that they found entire tuples, but here we find options of tuples. 

We now run our algorithm of section \ref{sec-algo-pfd} on this modified relation. Whenever we encounter the $\circledcirc$ symbol, we treat it as if the tuple does not exist. Whenever we encounter the $\phi$ symbol, we treat it as if it violates the current set of rules. The resulting confidence is the confidence of the CpFD.

We illustrate this process with an example. Consider the database of Figure \ref{fig-pafd-importance}. Consider the CpFD $(Color \rightarrow Type, T_1)$, where $T_1$ is the single-tuple $(Red, Star)$. We find that any tuple that does not have $Red$ for the $Color$ attribute matches the first query and is replaced with $\circledcirc$. The only two remaining options are {\it (Sirius, Red, Star)} and {\it (Taurus, Red, Nebula)}. since $Nebula \nasymp Star$, the {\it (Taurus, Red, Nebula)} option matches the second query, and thus the tuple is replaced with a $\phi$. We then run the pFD algorithm and obtain the confidence 0.5.

\subsection{Assessing the confidence of a CpAFD}
\label{sec-algo-cpafd}
We follow the same principle as Section \ref{sec-algo-cpfd}. However, in this case we use the Monte Carlo algorithm of Section \ref{sec-algo-pafd} instead of the algorithm for pFD to evaluate the confidence over the resulting relation.

\begin{figure*}[ht]

\begin{minipage}[b]{0.45\linewidth}
\vspace{20pt}
\centering
\includegraphics[keepaspectratio, width=240pt, clip, trim=0pt 140pt 40pt 0pt] {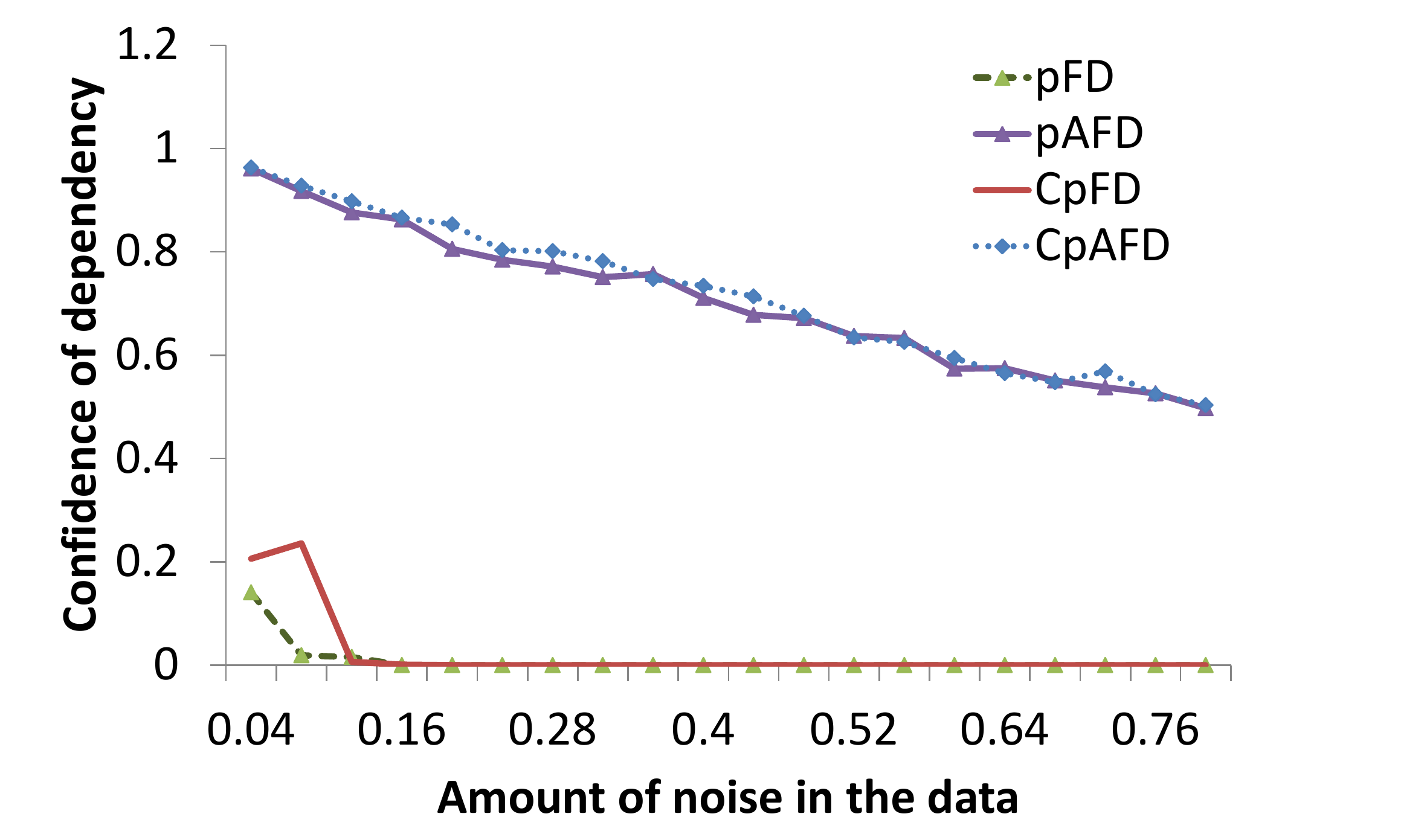}
\vspace{-0.2in}
\caption{Comparison between the average confidence reported for the dependencies in a database for varying noise.}
\vspace{-0.2in}
\label{fig-synth-conf-vs-noise}
\end{minipage}
\hspace{20pt}
\begin{minipage}[b]{0.45\linewidth}
\centering
\includegraphics[keepaspectratio, width=240pt, clip, trim=0pt 140pt 40pt 0pt] {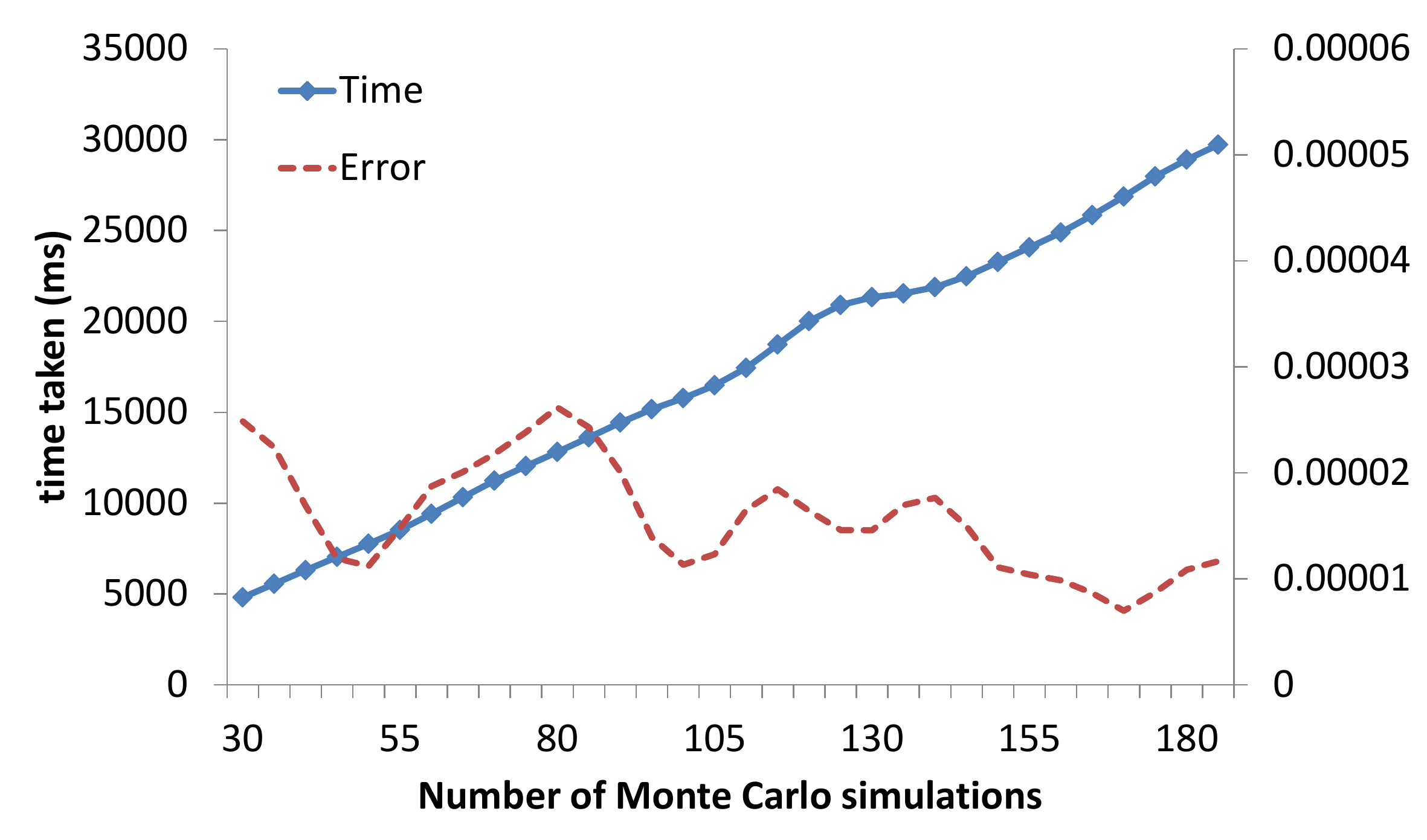}
\vspace{-0.2in}
\caption{The average error and the time taken vs the number of Monte Carlo simulations for a 200,000 tuple database of DBLP data.}
\vspace{-0.2in}
\label{fig-real-time-n-error-vs-sims}
\end{minipage}
\end{figure*}

\section{Experimental Evaluation}
\label{sec-experiments}

We empirically verify our algorithms using two sets of data -- one generated synthetically, and the other real data extracted from DBLP.

\subsection{Synthetic Data}
\label{sec-synthetic-data}
{\bf Data:} We are using the tuple-disjoint independent model. We generate synthetic data by creating a TDI database with $n$ tuples, each tuple having 2 options, and the options generated from the domain of the attributes of cardinality $m$. Each option may choose not to follow the specified FD with the probability called $noise$.

{\bf Results:} We evaluate the time taken by the pFD algorithm to run on synthetic data, while varying the number of tuples in the data. We compare this performance with the na\"ive algorithm, which would enumerate all the possible worlds. The results are shown in Figure \ref{fig-pfd-naive-vs-prune}. As can be clearly seen from the graph, the time taken by the na\"ive algorithm grows exponentially, and quickly becomes infeasible to run, even with as few as 30 tuples. 

In Figure \ref{fig-synth-time-vs-tuples}, we compare the time taken by the algorithms for the calculating the confidence of various dependencies. The cardinality of the domain of the attribute is held constant, while the number of tuples is increased. The Monte Carlo approximation algorithms for pAFD and CpAFD take significantly less time to converge compared to the pFD and CpFD algorithms. When plotted on a separate graph, it can be seen that they grow approximately linearly with the number of tuples. The pFD and CpFD show a general increasing trend with the number of tuples. The fluctuation observed is due to the algorithm quickly finding conflicting data and terminating early in some cases.

We also assessed the robustness of the dependencies to noise in the data in Figure \ref{fig-synth-conf-vs-noise}. We observe that with slight introduction of corruption, the confidences of a pFD and CpFD drop sharply. The confidence of a pAFD does not fall too sharply, which shows us that when the data is likely to be noisy, pAFD should be used to mine dependencies. However, for  data rectification, or in cases where the data is guaranteed to be clean, pFD will come in very useful, since the confidence value is very sensitive. 


\pagebreak
\subsection{DBLP Data}
\label{sec-real-data}
We use a set of data modified from \cite{kimura2010upi}. The database consists of DBLP \cite{ley2009dblp} data, with additional probabilistic attributes added to it by various information retrieval and machine learning sources. We use the ``Author" relation from this source, which contains information about approximately 700,000 computer science authors. This table has some deterministic attributes such as {\it Name, MinYearOfPublication, MaxYearOfPublication, NumPublication}. It also has the following uncertain attributes: {\it Institute, Country, Domain, Region and Subregion}.  We modify this dataset by re-indexing it and converting it into a tuple-disjoint independent format.

In Figure \ref{fig-real-time-n-error-vs-sims} we show the results of running the Monte Carlo pAFD algorithm on a 200,000-tuple subset of this dataset to evaluate the confidence of the dependency {\it Institute} $\rightsquigarrow Country$. We show how the accuracy of the evaluated confidence varies with the number of Monte Carlo simulations. We see that with increase in the number of simulations, the time taken increases, and the average error decreases, as expected. We terminate the simulation once the computed confidence converges. From this graph it is apparent that the it takes only around 100 simulations before the value stabilizes and the algorithm can be terminated.

In Figure \ref{fig-pafd-mc-vs-union} we show the confidence values of the pAFD mined from this data for all 700,000 tuples using two different approaches. As shown in Section \ref{sec-algo-pafd}, we can approximate the the confidence of a pAFD using three different approaches - the deterministic, union and Monte Carlo approximations. This experiment shows that even the union approximation method gives significant differences in confidence values. In the context of mining the dependencies we would typically choose dependencies by placing a threshold on the confidence values or by taking the top-k mined dependencies. The difference we observed in the confidence values is significant enough that the union method would give different dependencies when mining. We also see that the unioned method can both underestimate (e.g. the first dependency in Figure \ref{fig-pafd-mc-vs-union}) and overestimate the probabilities (e.g. the last dependency). Thus it seems that the Monte Carlo method is most suited to finding the confidence of pAFDs.

\begin{figure*}[ht]
\begin{minipage}[b]{0.3\linewidth}
\centering
\includegraphics[keepaspectratio, width=160pt, clip, trim=0pt 330pt 420pt 0pt] {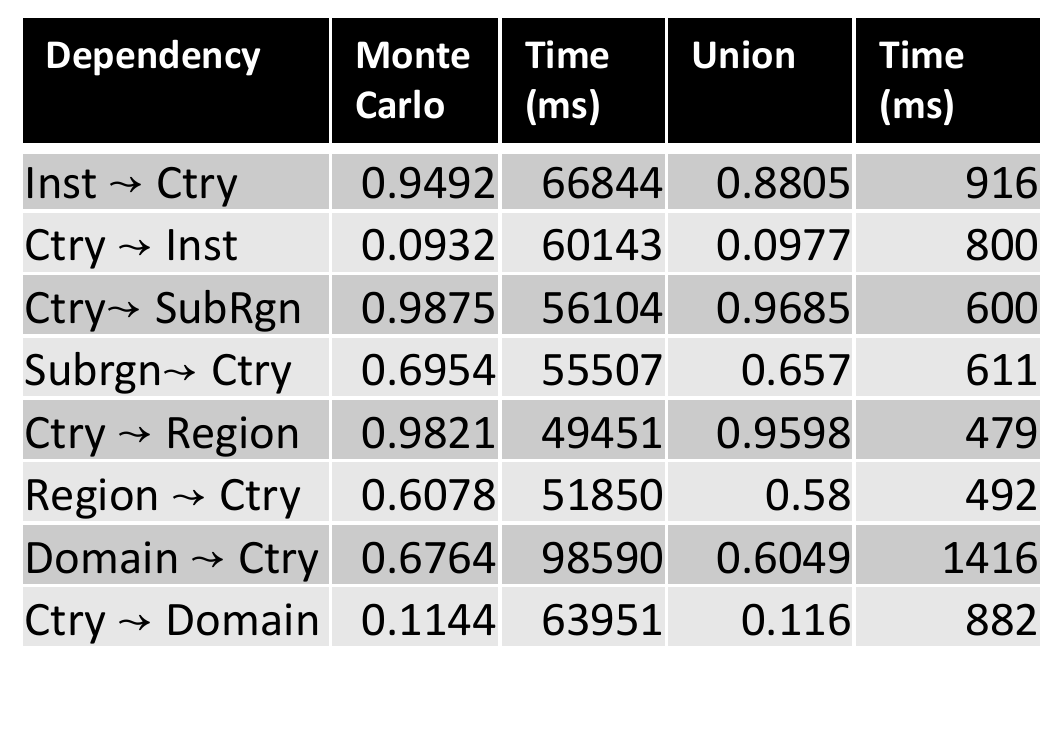}
\vspace{-0.2in}
\caption{The confidence of pAFD and time taken as computed by Monte Carlo method vs the union method.}
\vspace{-0.2in}
\label{fig-pafd-mc-vs-union}
\end{minipage}
\hspace{20pt}
\begin{minipage}[b]{0.3\linewidth}
\centering
\includegraphics[keepaspectratio, width=140pt, clip, trim=0pt 280pt 420pt 0pt] {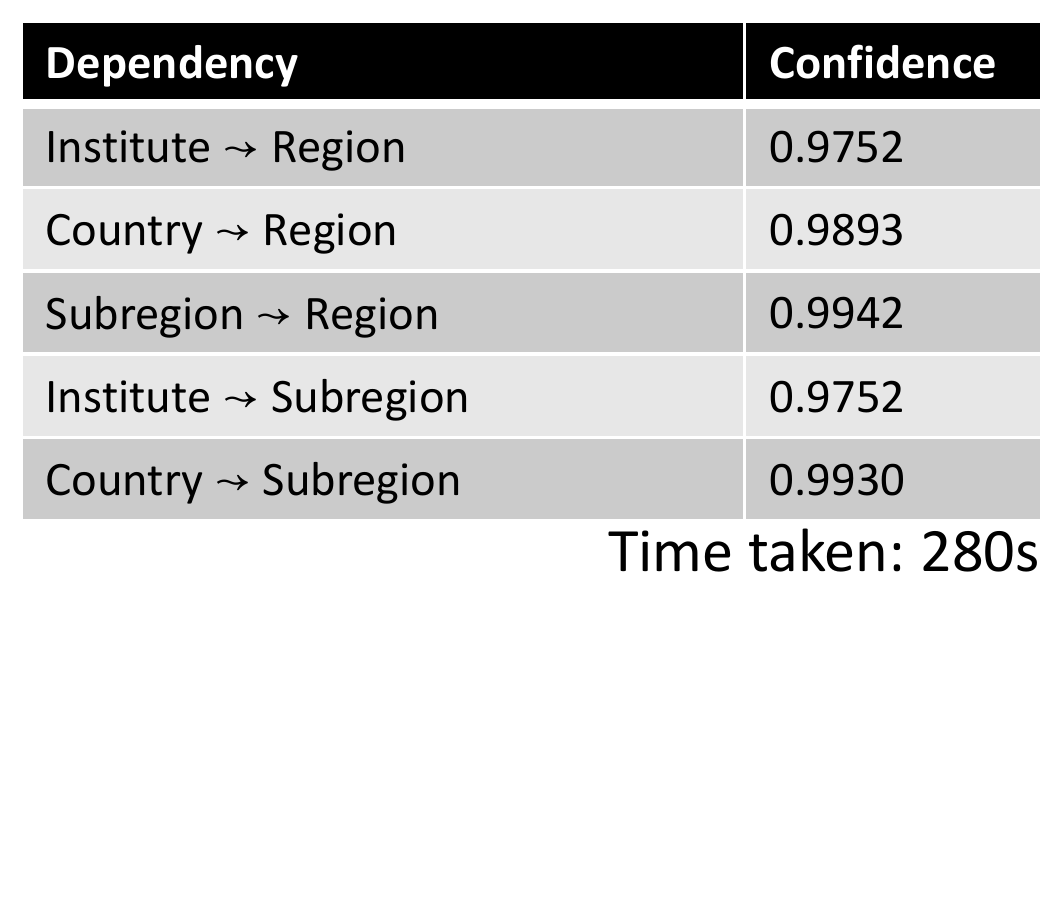}
\vspace{-0.2in}
\caption{The dependencies discovered in DBLP data by mining pAFDs for specificity threshold = 0.3.}
\vspace{-0.2in}
\label{fig-mined-dep-3}
\end{minipage}
\hspace{20pt}
\begin{minipage}[b]{0.3\linewidth}
\centering
\includegraphics[keepaspectratio, width=140pt, clip, trim=0pt 280pt 420pt 0pt] {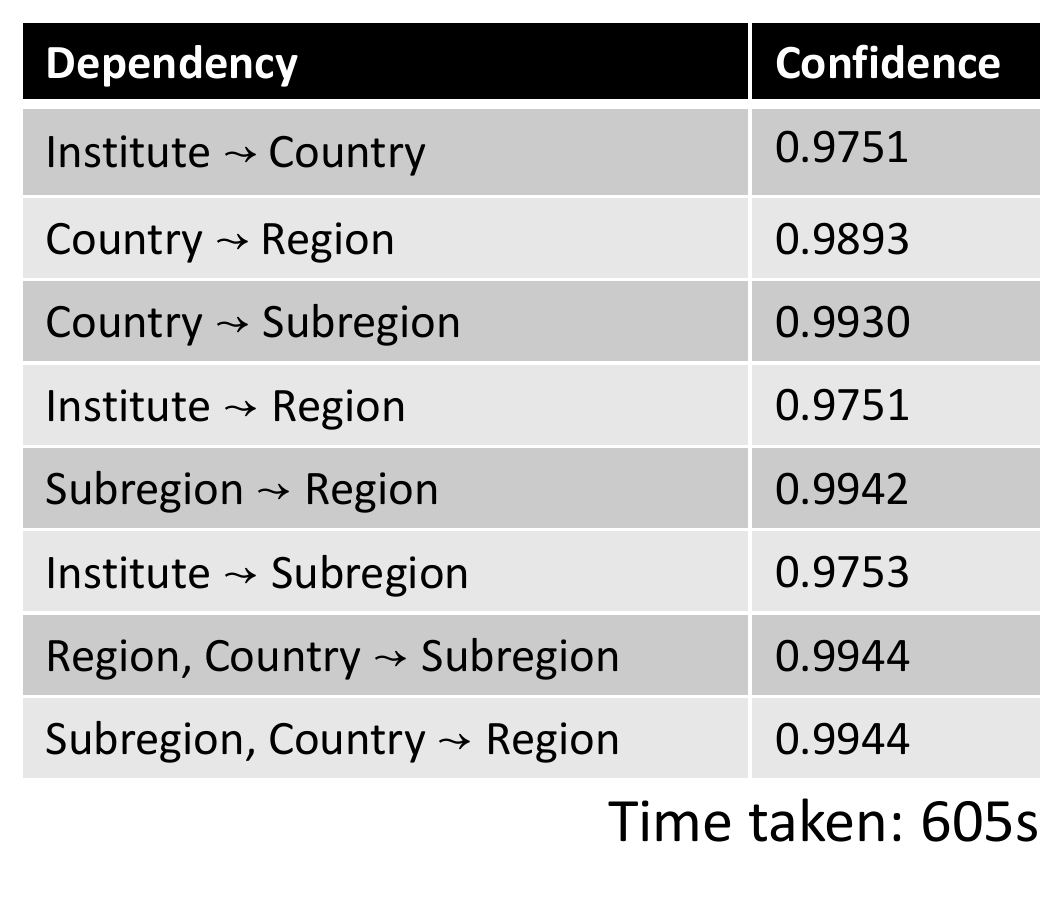}
\vspace{-0.2in}
\caption{The dependencies discovered in DBLP data by mining pAFDs with specificity threshold = 0.6.}
\vspace{-0.2in}
\label{fig-mined-dep-6}
\end{minipage}
\end{figure*}


\pagebreak
\subsection{Dependency mining}
\label{sec-dep-mining}

In order to mine pAFD dependencies, we adapted the $AFDMiner$ algorithm as described in \cite{wolf2009query}. We start with 200,000 tuples from the DBLP dataset. We choose four uncertain attributes {\it Institute, Country, Region and Subregion}. We build a heirarchy of possible dependencies (the set-con\-tain\-ment lattice of the attributes). Using the $AFDMiner$ algorithm, we prune our search space using two criteria: the redundancy and specificity. The redundancy condition prunes those dependencies that are guaranteed to hold since they are subsumed by the current set of dependencies. The specificity condition prunes those dependencies that are too specific to be considered as pAFDs. For each candidate dependency that did not get pruned off, we compute the confidence using the Monte Carlo algorithm. The exact details of the adaptation is described in the Appendix \ref{apx-afdminer}.

Figures \ref{fig-mined-dep-3} and \ref{fig-mined-dep-6} show dependencies mined from the DBLP data using two different thresholds for specificity. In the first case we set a very low specificity requirement, causing us to discover only those dependencies that are likely to be more general across the data. In the second case we set a high specificity threshold value, allowing much more specific dependencies to be also discovered. As we can see, this results in more dependencies being discovered, at the cost of quality of the dependencies.

\section{Related work}
\label{sec-relatedwork}

There is a large body of research that talks about association rules \cite{agrawal-fast} and itemsets \cite{brin1997dynamic}, more commonly known as the market-basket analysis problem. This work on association rules was recently improved by Kalavagattu \cite{aravind2008thesis} to include pruning based on specificity and to roll them up into approximate functional dependencies. AFDs have also been used to mine attribute correlations on autonomous web databases by Wolf {\it et al.} \cite{wolf2009query}. They have also been used by Wang {\it et al.} \cite{Wang2009} to find dirty data sources and normalize large mediated schemas.
FDs have also been generalized into conditional functional dependencies. Their role in data cleaning was shown by Bohannon {\it et al.} \cite{bohannon2007conditional}.

Sarma {\it et al.} extended FDs to probabilistic data in \cite{sarma2009schema}, however, in that paper the dependencies that were proposed were appropriate for schema normalization, but were inappropriate for discovering hidden relationships in data. Specifically, the horizontal dependencies specified can detect databases where the FD holds either in the union of all probabilistic tuples, each tuple individually, or within a specific tuple. The first two of these cases are intolerant to any noise in the data. The last one needs a single tuple to be specified, which is not holistic enough to discover any patterns. Such dependencies are ideal for schema normalization, since they allow the tables to be decomposed and simpler schema to be built, but it is not appropriate for discovering data patterns which needs to be fault tolerant.

Monte Carlo methods have been used in probabilistic databases before, for example, \cite{dalvi2007efficient} uses Monte Carlo methods to give top-k results for queries on probabilistic databases. A more general framework for probabilistic databases is proposed by Jampani {\it et. al.} in \cite{jampani2008mcdb} where the uncertainty is represented by parameters instead of probabilities, so that a more generalized model of uncertainty can be represented.

In \cite{gupta2006creating}, Gupta and Sarawagi demonstrate how to create probabilistic databases that are an approximation of an information extraction model and find that using the appropriate model of uncertainty in a database is important. If the model of uncertainty  is  too simple then interactions between elements of the generating model cannot be represented; if it is too complex then querying becomes inefficient. Similarly, in this paper, we are proposing the right level of uncertainty, but for functional dependencies. We show that using probabilistic semantics does cause a significant change in the confidence of the dependencies, and we show efficient algorithms that find these dependencies.

\section{Conclusions}
\label{sec-conclusions}
In this paper we defined a spectrum dependencies for probabilistic databases. These dependencies are logical extensions of their deterministic counterparts. We explained how these dependencies are related to each other. We showed that pAFD would always have a larger confidence than the pFD. We showed the CpAFDs were the most general of all and that it subsumed every other kind of dependency. We then presented algorithms to assess the confidence of each of these dependencies. We empirically verified the algorithms -- the ones for pFD and CpFD were exponential in the number of values of the attribute, and approximately linear in the number of tuples. The Monte Carlo algorithms for the approximate dependencies converged fast and were accurate. We also showed experiments with real data that demonstrated that the Monte Carlo algorithm converges quickly. Finally we showed how we can use these algorithms to effectively mine dependencies from a real probabilistic database and discover useful dependencies. We are currently exploring the use of these dependencies in the {\sc qpiad} project \cite{wolf2009query}.


\bibliographystyle{abbrv}
{
\small\bibliography{pdb}  
}
\pagebreak
\section{Appendix}
\label{sec-appendix}

\subsection{Adapting the algorithms for TI databases}
\label{apx-algo-for-ti}
We can use the pFD algorithm for TI databases with a slight modification. Recall the symbol $\circledcirc$ introduced in Section \ref{sec-algo-cpfd} for handling options eliminated for not matching the pattern tableau. The symbol represents that in further computation, the option will be ignored for the purpose of computing the confidence, that is, it will not conflict with any existing rule. For every tuple in the TI database whose probability $p$ is less than 1, we convert it into a TDI database by adding an option to it consisting of the $\circledcirc$ symbol and probability $1-p$. We then have a TDI database which is essentially equivalent to the TI database. We can now apply the pFD algorithm on this database to assess its confidence.

The union approximation for assessing the confidence of a pAFD from Section \ref{sec-algo-pafd} can be applied directly to the TI database to get the accurate value of the pAFD.

\subsection{Adapting specificity for probabilistic databases}
\label{apx-specificity}
In this section we will first introduce the notion of specificity as described by Kalavagattu and Wolf {\it et al.} \cite{aravind2008thesis, wolf2009query} for deterministic databases and then show how it is adapted to probabilistic databases.

{\bf Deterministic databases:} The distribution of values for the determining set is an important measure to judge the ``usefulness" of an AFD. For an AFD $X \rightsquigarrow A$, having fewer distinct values of $X$ means that there exist more tuples in the database that have the same values of $X$. This makes the AFD potentially more relevant. This is because if every value of $X$ is distinct (i.e. it is a key) then the AFD trivially holds; however if the dependency holds in spite of $X$ having only a few distinct values, the AFD has a deeper semantic meaning.

To quantify this, we first define  the {\em support of a value} $\alpha_{i}$ of an attribute set $X$, $support(\alpha_{i})$, as the occurrence frequency of value $\alpha_{i}$ in the training set. The support is defined as 
$support (\alpha_{i}) =  count(\alpha_{i})/N,$ where $N$ is the number of tuples in the training set.

Now we measure how the values of an attribute set $X$ are
distributed using {\it specificity}. Specificity is defined as the
information entropy of the set of all possible values of attribute
set $X$:
 \{$\alpha_{1}$, $\alpha_{2}$, \ldots, $\alpha_{m}$ \},
normalized by the maximal possible  entropy (which is achieved when
$X$ is a key). Thus, specificity is a value that lies between 0 and
1.

\vspace*{-0.2in}
\begin{eqnarray}
 \emph{specificity ($X$)}
&=& \frac{- \sum_{1}^{m}
support(\alpha_{i})\times\log_{2}(support(\alpha_{i}))}{\log_{2}(\emph{N})}
\nonumber
\end{eqnarray}

When there is only one possible value of $X$, then this value has
the maximum support and is the least specific, thus we have
specificity equals to 0.
When there are many distinct values in $X$, each having a low support and are specific,
we have a high value of  specificity.
When all values of $X$ are distinct (when $X$ is a key), each value has the minimum support and is most specific and has specificity equal to 1.

Now we overload the concept of specificity on AFDs. The specificity of
an AFD is defined as the specificity of its determining set. i.e. 
$\emph{specificity ($X \rightsquigarrow A$)}  = \emph{specificity ($X$)}$.
The lower specificity  of an AFD, potentially the more  relevant
possible answers can be retrieved using the rewritten queries
generated by this AFD, and thus a higher recall for a given number
of rewritten queries.

Intuitively, specificity increases when the number of distinct
values for a set of attributes increases. Consider two attribute sets $X$ and $Y$ such that Y$\supset$X. Since $Y$ has more
attributes than $X$, the number of distinct values of $Y$ is no less
than that of $X$, specificity($Y$) is no less than specificity($X$).

{\bf Probabilistic Databases:} In a probabilistic database, the specificity of an attribute set $X$ would be defined as the weighted average of the specificity of $X$ in each possible world. Computing this is potentially exponential in the number of tuples, since every possible world will have a different set of association rules with different support.

In this paper, we are using specificity to prune our search space. We need to be able to compute the specificity very quickly so that we do not spend too much time deciding whether or not to prune the current subspace of dependencies. As a result, we decide to approximate the computation of specificity by using a method similar to the union method described in Section \ref{sec-algo-pafd}. We ignore the intra-tuple correlations, and create a TI database by taking the union of all the options in our TDI database. Computing the specificity of a TI database is a straightforward adaptation of the deterministic algorithm. The definition for  specificity($X$) remains the same, but we redefine $support(\alpha_i)$ as (where $t$ represents all the tuples in the TI database):
\[support(\alpha_i) = \displaystyle\sum_{\alpha_i \in t}{prob(t)}/\sum{prob(t)}\]

\subsection{Adaptations to AFDMiner}
\label{apx-afdminer}
We adapt the {\it AFDMiner} algorithm from Wolf {\it et al.} \cite{wolf2009query} to mine dependencies in our data. In this section we describe the outline of the algorithm, and the adaptations to probabilistic data.

The algorithm searches through the set-containment lattice of the attributes of the relation. This lattice consists of all possible sets of attributes. Each set of attribute has a directed edge that points to all sets that contain one attribute more than itself. The algorithm performs a breadth-first search through this lattice, starting with the null set of attributes and working its way up to the set of all attributes. For each directed edge $(X, X \cup \{A\})$ the algorithm travels along, the dependencies $(X \rightsquigarrow A)$ are tested. {\it AFDMiner} outputs those dependencies whose confidence is larger than the supplied confidence threshold. We adapt {\it AFDMiner} by supplying our own confidence assessing functions for pAFD.

{\em Pruning:} Each attribute set $X$ is tested for its specificity value. If the value is higher than the specificity threshold, then all outgoing edges from that set are removed from the lattice. This lets the algorithm prune the space of dependencies whose body is $X$ or its superset, since they are guaranteed to be above the specificity threshold. We use the specificity definition from Section \ref{apx-specificity} for this algorithm.

AFDMiner further prunes the space of dependencies based on redundancy. For any dependencies that hold exactly, any superset of the dependency would also hold (by Armstrong's Axioms), and hence need not to be checked. So, effectively, a list of exact FDs is maintained, and any superset of the FDs in this list are not checked. In our case, the algorithm to check for a pFD is significantly more expensive than the algorithm to compute the confidence of a pAFD. As a result, we adapt this condition by replacing FDs with {\em very high confidence} pAFDs. For any pAFD that has a confidence larger than the preset high-confidence threshold, we prune the outgoing edges from that attribute.

\balancecolumns
\end{document}